\documentclass[twoside,leqno]{article}

\usepackage{ltexpprt}
\usepackage{hyperref}

\usepackage{amsfonts, amstext}
 \usepackage{libertine}
 \usepackage[libertine]{newtxmath}

\usepackage[font={small,it}]{caption}

\usepackage{setspace}

\usepackage{enumitem}

\usepackage[nameinlink]{cleveref}
\Crefname{@theorem}{Theorem}{Theorems}
\Crefname{algocf}{Algorithm}{Algorithms}
\crefname{algocfline}{line}{lines}
\Crefname{invariant}{Invariant}{Invariants}
\Crefname{claim}{Claim}{Claims}
\Crefname{subclaim}{Subclaim}{Subclaims}

\usepackage{epsfig}
\usepackage{xspace}
\usepackage{soul}
\usepackage{latexsym}
\usepackage{bbm}

\usepackage{framed}

\usepackage[dvipsnames]{xcolor}
\definecolor{DarkGray}{rgb}{0.66, 0.66, 0.66}
\definecolor{DarkPowderBlue}{rgb}{0.0, 0.2, 0.6}
\definecolor{fluorescentyellow}{rgb}{0.8, 1.0, 0.0}

\usepackage{thmtools,thm-restate}

\usepackage{tikz}
\usetikzlibrary{shapes.geometric}
\usepackage{subcaption}

\newcounter{note}[section]

\sethlcolor{fluorescentyellow}

\newcommand{\initOneLiners}{%
    \setlength{\itemsep}{0pt}
    \setlength{\parsep }{0pt}
    \setlength{\topsep }{0pt}
}

\newcommand{\tO}{\tilde{O}}


\newcommand{\lca}{\operatorname{lca}}
\newcommand{\poly}{\operatorname{poly}}

\newcommand{\polylog}{\operatorname{polylog}}
\renewcommand{\emptyset}{\varnothing}

\newcommand{\junk}[1]{}
\newcommand{\eat}[1]{}

\newif\ifhideproofs

\ifhideproofs
\usepackage{environ}
\NewEnviron{hide}{}

\fi

\def\ShowComment{True}

\ifdefined\ShowComment
\def\ruoxu#1{\marginpar{$\leftarrow$\fbox{R}}\footnote{$\Rightarrow$~{\sf\textcolor{cyan}{#1 --Ruoxu}}}}

\else
\def\ruoxu#1{}
\fi 

\newcommand{\ext}[1]{T^{\rm ext}_{#1}}
\newcommand{\ct}{{\sc Cut Threshold}\xspace}
\newcommand{\con}{\lambda}
\newcommand{\Con}{\Lambda}
\newcommand{\cut}{{\rm ct}}
\newcommand{\cnt}{{\rm cnt}}
\newcommand{\gen}{{\rm gen}}

\newcommand{\dc}{DECA\xspace}
\newcommand{\dem}{\text{dem}}

\begin{document}

\title{Augmenting Edge Connectivity via Isolating Cuts}
\author{Ruoxu Cen\thanks{Department of Computer Science, Duke University.  Email: {\tt ruoxu.cen@duke.edu}}\and Jason Li\thanks{Simons Institute, UC Berkeley. Email: {\tt jmli@cs.cmu.edu}}\and Debmalya Panigrahi\thanks{Department of Computer Science, Duke University. Email: {\tt debmalya@cs.duke.edu}}}
\date{}

\maketitle


\begin{abstract}
    We give an algorithm for augmenting the edge connectivity of an undirected graph by using the {\em isolating cuts} framework (Li and Panigrahi, FOCS '20). Our algorithm uses poly-logarithmic calls to any max-flow algorithm, which yields a running time of $\tilde O(m + n^{3/2})$ and improves on the previous best time of $\tilde O(n^2)$ (Bencz\'ur and Karger, SODA '98) for this problem. We also obtain an identical improvement in the running time of the closely related edge splitting off problem in undirected graphs.
\end{abstract}




\section{Introduction}
\label{sec:intro}

In the {\em edge connectivity augmentation} problem, we are given an undirected graph $G = (V, E)$ with (integer) edge weights $w$, and a target connectivity $\tau > 0$. The goal is to find a minimum weight set $F$ of edges on $V$ such that adding these edges to $G$ makes the graph $\tau$-connected. (In other words, the value of the minimum cut of the graph after the augmentation should be at least $\tau$.) The edge connectivity augmentation problem is known to be tractable in $\poly(m, n)$ time, where $m$ and $n$ denote the number of edges and vertices respectively in $G$. This was first shown by Watanabe and Nakamura~\cite{WatanabeN87} for unweighted graphs, and the first strongly polynomial algorithm was obtained by Frank~\cite{Frank92}. Since then, several algorithms~\cite{CaiS89,NaorGM97,Gabow16,Gabow94,NagamochiI97} have progressively improved the running time to the current best $\tO(n^2)$ obtained by Bencz\'ur and Karger~\cite{BenczurK00}.\footnote{$\tO(\cdot)$ ignores (poly)-logarithmic factors in the running time.}
In this paper, we give an algorithm to solve the edge connectivity augmentation problem using $\polylog(n)$ calls to {\em any} max-flow algorithm:
\begin{theorem}\label{thm:augment}
    There is a randomized, Monte Carlo algorithm for the edge connectivity augmentation problem that runs in $\tO(m) + \polylog(n)\cdot F(m, n)$ time where $F(m, n)$ is the running time of any maximum flow algorithm on an undirected graph containing $m$ edges and $n$ vertices.
\end{theorem}
Using the current best max-flow algorithm on undirected graphs~\cite{BrandLLSSSW21},\footnote{We note that for sparse graphs, there is a slightly faster max-flow algorithm that runs in $O(m^{3/2-\delta})$ time~\cite{GaoLP21}, where $\delta > 0$ is a small constant. If we use this max-flow algorithm in \Cref{thm:augment}, we also get a running time of $O(m^{3/2-\delta})$ for the augmentation problem.} this yields a running time of $\tO(m + n^{3/2})$, thereby improving on the previous best bound of $\tO(n^2)$.

The edge connectivity augmentation problem is closely related to {\em edge splitting off}, a widely used tool in the graph connectivity literature (e.g.,~\cite{Gabow94,NagamochiI97}). A pair of (weighted) edges $(u, s)$ and $(s, v)$ both incident on a common vertex $s$ is said to be split off by weight $w$ if we reduce the weight of both these edges by $w$ and increase the weight of their {\em shortcut} edge $(u, v)$ by $w$. Such a splitting off is valid if it does not change the (Steiner) connectivity of the vertices $V\setminus \{s\}$. 
If all edges incident on $s$ are eliminated by a sequence of splitting off operations, we say that the vertex $s$ is split off. We call the problem of finding a set of edges to split off a given vertex $s$ the edge splitting off problem. 

Lov\'asz~\cite{Lovasz79} initiated the study of edge splitting off by showing that any vertex $s$ with even degree in an undirected graph can be split off while maintaining the (Steiner) connectivity of the remaining vertices. (Later, more powerful splitting off theorems~\cite{Mader78} were obtained that preserve stronger properties and/or apply to directed graphs, but these come at the cost of slower algorithms. We do not consider these extensions in this paper.) The splitting off operation has emerged as an important inductive tool in the graph connectivity literature, leading to many algorithms with progressively faster running times being proposed for the edge splitting off problem~\cite{CaiS89,Frank92,Gabow94,NagamochiI97}. Currently, the best running time is $\tO(n^2)$, which was obtained in the same paper of Bencz\'ur and Karger that obtained the edge connectivity augmentation result~\cite{BenczurK00}. We improve this bound as well:
\begin{theorem}\label{thm:splitting}
    There is a randomized, Monte Carlo algorithm for the edge splitting off problem that runs in $\tO(m) + \polylog(n)\cdot F(m, n)$ time where $F(m, n)$ is the running time of any maximum flow algorithm on an undirected graph containing $m$ edges and $n$ vertices.
\end{theorem}

As in previous work (e.g.,~\cite{BenczurK00}), instead of giving separate algorithms for the edge connectivity augmentation and the edge splitting off problems, we give an algorithm for the {\em degree-constrained} edge connectivity augmentation (\dc) problem, which generalizes both these problems. In this problem, given an edge connectivity augmentation instance, we add additional {\em degree constraints} $\beta(v)\ge 0$ requiring the total weight of added edges incident on each vertex to be at most its degree constraint. The goal is to either return an optimal set of edges for the augmentation problem that satisfy the degree constraints, or to say that the instance is infeasible. 

Clearly, \dc generalizes the edge connectivity augmentation problem. To see why \dc also generalizes splitting off, create the following \dc instance from a splitting off instance: Remove the edges incident on $s$ and set $\beta(v)$ to the weighted degree of $v$ in these edges. Then, set $\tau$ to the (Steiner) connectivity of $V$ in the input graph. Once the \dc solution $F$ is obtained, for vertices $v$ whose degree in $F$ is smaller than $\beta(v)$, use an arbitrary weighted matching to increase the degrees to exactly $\beta(v)$.

For the \dc problem, we show that:
\begin{theorem}\label{thm:deg-augment}
    There is a randomized, Monte Carlo algorithm for the degree-constrained edge connectivity augmentation problem that runs in $\tO(m) + \polylog(n)\cdot F(m, n)$ time where $F(m, n)$ is the running time of any maximum flow algorithm on an undirected graph containing $m$ edges and $n$ vertices.    
\end{theorem}
\Cref{thm:deg-augment} implies \Cref{thm:augment} and \Cref{thm:splitting}. 
The rest of this paper focuses on proving \Cref{thm:deg-augment}.

\subsection{Our Techniques}
A key tool in many augmentation/splitting off algorithms (e.g., in \cite{WatanabeN87,NaorGM97,Gabow16,Benczur94,BenczurK00}) is that of {\em extreme sets}. A non-empty set of vertices $X\subset V$ is called an extreme set in graph $G = (V, E)$ if for every proper subset $Y\subset X$, we have $\delta_G(Y) > \delta_G(X)$, where $\delta_G(X)$ (resp., $\delta_G(Y)$) is the total weight of edges with exactly one endpoint in $X$ (resp., $Y$) in $G$. (If the graph is unambiguous, we drop the subscript $G$ and write $\delta(\cdot)$.) The extreme sets form a laminar family, therefore allowing an $O(n)$-sized representation in the form of an {\em extreme sets tree}. The main bottleneck of the Bencz\'ur-Karger algorithm is in the construction of the extreme sets tree. They use the {\em recursive contraction} framework of Karger and Stein~\cite{KargerS96} for this construction, which takes $\tO(n^2)$ time. In this paper, we obtain a faster algorithm for finding the extreme sets of a graph:
\begin{restatable}{theorem}{Extreme}\label{thm:extreme}
    There is a randomized, Monte Carlo algorithm for finding the extreme sets tree of an undirected graph that runs in $\tO(m) + \polylog(n)\cdot F(m, n)$ time where $F(m, n)$ is the running time of any maximum flow algorithm on an undirected graph containing $m$ edges and $n$ vertices.
\end{restatable}

Our extreme sets algorithm is based on the {\em isolating cuts} framework that we introduced in a recent paper~\cite{LiP20deterministic}. (This was independently discovered by Abboud~{\em et al.}~\cite{AbboudKT21}.) Given a set of $k$ terminal vertices, this framework uses $O(\log k)$ max-flows to find the minimum cuts that separate each individual terminal from the remaining terminals (called isolating cuts). In the current paper, instead of using the framework directly, we use a gadget called a \ct that is defined as follows: for a given vertex $s$ and threshold $\phi\ge 0$, the \ct $\cut(s, \phi)$ is the set of vertices $t$ such that the value of the minimum $s-t$ cut $\con(s, t) \le \phi$. We showed recently~\cite{LiP21approximate} that the isolating cuts framework can be used to find the \ct for any vertex $s$ and threshold $\phi$ in $\polylog(n)$ max-flows. We use this result here, and focus on obtaining extreme sets using a \ct subroutine.

Our main observation is that if an extreme set $Y$ partially overlaps the {\em complement} $X$ of a \ct, then it must actually be wholly contained in $X$. (Intuitively, one may interpret this property as saying that an extreme set and a \ct are {\em non-crossing}, although our property is actually stronger, and only the non-crossing property does not suffice for our algorithm.) This allows us to design a divide and conquer algorithm that runs a recursion on two subproblems generated by contracting each side of a carefully chosen \ct. The above property ensures that every extreme set in the original problem continues to be an extreme set in either of the two subproblems. In order to bound the depth of recursion, it is important to use a \ct that produces a {\em balanced} partition of vertices. We ensure this by adapting a recent observation of Abboud {\em et al.}~\cite{AbboudKT20} which asserts that a \ct based on the connectivity between two randomly chosen vertices is balanced with constant probability. One additional complication is that while the contraction of the \ct (or its complement) does not eliminate any extreme set, it might actually add new extreme sets. We run a {\em post-processing} phase where we use a dynamic tree data structure to eliminate these spurious extreme sets added by the recursive algorithm.

After obtaining the extreme sets tree, the next step (in our algorithm and in previous work such as \cite{BenczurK00}) is to add a vertex $s$ and use a postorder traversal on the extreme sets tree to find an optimal set of edges incident on $s$ for edge connectivity augmentation. This step takes $O(n)$ time. 

Next, we split off vertex $s$ using an iterative algorithm that again uses the extreme sets tree. At a high level, this splitting off algorithm follows a similar structure to the Bencz\'ur-Karger algorithm, but with a couple of crucial differences that improves the running time from $\tO(n^2)$ to $\tO(m)$. The first difference is in the construction of a min-cut {\em cactus} data structure. At the time of the Bencz\'ur-Karger result, the fastest cactus algorithm was based on recursive contraction~\cite{KargerS96} and had a running time of $\tO(n^2)$. But, this has since been improved to $\tO(m)$ by Karger and Panigrahi~\cite{KargerP09}. Using this faster algorithm removes the first $\tO(n^2)$ bottleneck in the augmentation algorithm. 

The second and more significant improvement is in the use of data structures in the splitting off algorithm. This is an iterative algorithm that has $O(n)$ iterations and adds $O(n)$ edges in each iteration. The Bencz\'ur-Karger algorithm updates its data structures for each edge in all these iterations, thereby incurring $O(n^2)$ updates. Instead, we use the following observation (this was known earlier): there are only $O(n)$ {\em distinct} edges used across the $O(n)$ iterations, and the total number of {\em changes} in the set of edges from one iteration to the next is $O(n)$. To exploit this property, we use a {\em lazy} procedure based on the top tree data structure due to Goldberg {\em et al.}~\cite{goldberg1991use} (and additional priority queues to maintain various ordered lists). Our data structure only performs updates on edges that are added/removed in an iteration, thereby reducing the total number of updates to $O(n)$, and each update can be implemented in $O(\log n)$ using standard properties of top trees and priority queues. We obtain the following:

\begin{restatable}{theorem}{Deca}\label{thm:deca}
Given an input graph and its extreme set tree, there is an $\tilde O(m)$ time algorithm that solves the degree-constrained edge connectivity problem.
\end{restatable}

\Cref{thm:deg-augment} now follows from \Cref{thm:extreme} and \Cref{thm:deca}.

\paragraph{Roadmap.} We give the algorithm for finding extreme sets that establishes \Cref{thm:extreme} in \Cref{sec:extreme}. The algorithm for the \dc problem that uses the extreme sets tree and establishes \Cref{thm:deca} is given in \Cref{sec:augment}.

\section{Algorithm for Extreme Sets}
\label{sec:extreme}

In this section, we present our extreme sets algorithm and prove \Cref{thm:extreme}, restated below.
\Extreme*

Recall that the input graph $G = (V, E)$ is an undirected graph with integer edge weights $w$. An extreme set is a set of vertices $X \subset V$ such that for every proper subset $Y\subset X$, we have $\delta(Y) > \delta(X)$. 
Note that all singleton vertices are also extreme sets by default since they do not have non-empty strict subsets. 

The following is a well-known property of extreme sets (see, e.g.,~\cite{BenczurK00}):
\begin{lemma}\label{lem:extreme-laminar}
    The extreme sets of an undirected graph form a laminar family, i.e., for any two extreme sets, either one is contained in the other, or they are entirely disjoint.
\end{lemma}
This lemma allows us to represent the extreme sets of $G = (V, E)$ as a rooted tree $\ext{G}$ with the following properties:
\begin{itemize}
    \item 
    The set of vertices in $V$ exactly correspond to the set of leaf vertices in $\ext{G}$.
    \item
    The extreme sets in $G$ exactly correspond to the (proper) subtrees of $\ext{G}$ in the following sense: for any extreme set $X\subset V$, there is a unique subtree of $G$ denoted $\ext{G}(X)$ such that the vertices in $X$ are exactly the leaves in $\ext{G}(X)$. Overloading notation, we also use $\ext{G}(X)$ to denote the root of the subtree corresponding to $X$ in $\ext{G}$.
\end{itemize}  We call $\ext{G}$ the {\em extreme sets tree} of $G$, and give an algorithm to construct it in this section.

We will use a {\em \ct} procedure from our recent work~\cite{LiP21approximate}. Recall that a \ct is defined as follows:
\begin{Definition}
    Let $\con(s, t)$ denote the value of the max-flow between two vertices $s$ and $t$; we call $\con(s, t)$ the {\em connectivity} between $s$ and $t$.
    Then, the \ct for vertex $s$ and threshold $\phi \ge 0$, denoted $\cut(s, \phi)$, is the set of all vertices $t\in V\setminus \{s\}$ such that $\con(s, t) \le \phi$. 
\end{Definition}
In recent work, we gave an algorithm for finding a \ct~\cite{LiP21approximate} based on our isolating cuts framework~\cite{LiP20deterministic}:
\begin{theorem}[Li and Panigrahi~\cite{LiP21approximate}]\label{thm:ct-alg}
    Let $G = (V, E)$ be an undirected graph containing $m$ edges and $n$ vertices.
    For any given vertex $s\in V$ and threshold $\phi\ge 0$, there is a randomized Monte Carlo algorithm for finding the \ct $\cut(s, \phi)$ in $\tO(m) + \polylog(n)\cdot F(m, n)$ time, where $F(m, n)$ is the running time of any max-flow algorithm on undirected graphs containing $m$ edges and $n$ vertices.
\end{theorem}

In order to use this result, we first relate extreme sets to \ct. We need the following definition:
\begin{Definition}
We say that a set of vertices $X$ {\em respects} the extreme sets of $G$ if for any extreme set $Y$ of $G$, one of the following holds: (a) $Y\subseteq X$ or (b) $X\subseteq Y$ or (c) $X\cap Y = \emptyset$. In other words, if there exist two vertices $x_1, x_2\in X$ such that $x_1\in Y$ and $x_2\notin Y$, then it must be that $Y\subset X$. 
\end{Definition}
Our main observation that relates extreme sets to \ct is the following:
\begin{lemma}\label{lem:ct-noncross}
    Let $G = (V, E)$ be an undirected graph. For any vertex $s\in V$ and threshold $\phi\ge 0$, the complement of the \ct  $\cut(s, \phi)$, denoted $X := V\setminus \cut(s, \phi)$, respects the extreme sets of $G$.
\end{lemma}
Note that $s\in X$ by definition of $\cut(s, \phi)$. The crucial ingredient in the proof of \Cref{lem:ct-noncross} is that minimum $s-t$ cuts for any $t\notin X$ are non-crossing with respect to the cut $\cut(s, \phi)$:
\begin{lemma}\label{lem:contain}
    For any vertex $t\in \cut(s,\phi)$, the side containing $t$ of a minimum $s-t$ cut must be entirely contained in $\cut(s, \phi)$.
\end{lemma}
\begin{proof}
    Suppose not; then, there is at least one vertex $t'\notin \cut(s, \phi)$ such that the minimum $s-t$ cut also separates $s$ and $t'$. But, then $\con(s, t') \le \con(s, t) \le \phi$. This contradicts $t'\notin \cut(s, \phi)$.
\end{proof}
Now, we use \Cref{lem:contain} to prove \Cref{lem:ct-noncross}. 
\begin{proof}[Proof of \Cref{lem:ct-noncross}]
An extreme set $Y$ that violates \Cref{lem:ct-noncross} has the following properties: (a) $Y$ separates $s, t'$ for some vertex $t'\notin \cut(s, \phi)$, and (b) $Y$ contains some vertex $t\in \cut(s, \phi)$. Let $Z$ denote the side containing $t$ of a minimum $s-t$ cut. 

Now, since the cut function is submodular, we have:
\begin{equation}\label{eq:submodular}
    \delta(Y) + \delta(Z) \ge \delta(Y\cap Z) + \delta(Y\cup Z).
\end{equation}
But, by \Cref{lem:contain}, we have $Z\subseteq \cut(s, \phi)$. Now, since $Y$ separates $s, t'\notin \cut(s, \phi)$, if follows that $Y\cup Z$ also separates $s, t'$. As a consequence, 
\begin{equation}\label{eq:1}
    \delta(Y\cup Z) \ge \lambda(s, t') > \phi.    
\end{equation}
Finally, since $t\in \cut(s, \phi)$, we have $\lambda(s, t) \le \phi$. Since $Z$ is a minimum $s-t$ cut, it follows that:
\begin{equation}\label{eq:2}
    \delta(Z) \le \phi.
\end{equation}
Combining \Cref{eq:1} and \Cref{eq:2}, we get:
\begin{equation}\label{eq:3}
    \delta(Z) < \delta(Y\cup Z).
\end{equation}
Finally, we note $Y\cap Z$ is a {\em proper} subset of $Y$. This is because $Y$ contains one vertex among $s, t'\notin \cut(s, \phi)$ by virtue of separating them, but $Z$ is entirely contained in $\cut(S, \phi)$ by \Cref{lem:contain}. Now, since $Y$ is an extreme set, we have
\begin{equation}\label{eq:4}
    \delta(Y\cap Z) > \delta(Y).    
\end{equation}
The lemma follows by noting that \Cref{eq:3} and \Cref{eq:4} contradict \Cref{eq:submodular}.
\end{proof}

\subsection{Description of the Algorithm}

We now use \Cref{lem:ct-noncross} to design a divide and conquer algorithm for extreme sets. The algorithm has two phases. In the first phase, we construct a tree $T$ that includes all extreme sets of $G$ as subtrees, but might contain other subtrees that do not correspond to extreme sets. In the second phase, we remove all subtrees of $T$ that are not extreme sets and obtain the final extreme sets tree $\ext{G}$. 


\paragraph{Phase 1:}
The first phase of the algorithm uses a recursive {\em divide and conquer} strategy. A general recursive subproblem is defined on a graph $G^\gen  = (V^\gen, E^\gen)$ that is obtained by contracting some sets of vertices in $G$ that will be defined below. The contracted vertices are denoted $C^\gen$ and the uncontracted vertices $U^\gen\subseteq V$. Thus, $V^\gen = C^\gen \cup U^\gen$. Note that the contracted vertices $C^\gen$ form a partition of the vertices in $V\setminus U^\gen$. The graph $G^\gen$ is obtained from $G$ by contracting each set of vertices that is represented by a single contracted vertex in $C^\gen$, deleting self-loops and unifying parallel edges into a single edge whose weight is the cumulative weight of the parallel edges. The goal of the recursive subproblem on $G^\gen$ is to build a tree $T(G^\gen)$ that contains all extreme sets in $G^\gen$. Initially, $U^\gen = V$ and $C^\gen = \emptyset$, i.e., $G^\gen = G$. Therefore, the overall goal of the algorithm is to find all extreme sets of $G$.

First, we perturb the edge weights of the input graph $G^\gen$ as follows: We independently generate a random value $r(u, v)$ for each edge that is drawn from the uniform distribution defined on $\{1, 2, \ldots, N\}$. (We will set the precise value of $N$ later, but it will be polynomial in the size of the graph $G^\gen$.) We define new edge weights $w'(u, v) := mN\cdot w(u, v) + r(u, v)$ for all edges $(u, v)\in E$. We first show that all extreme sets under the original edge weights $w$ continue to be extreme sets under the new edge weights $w'$:
\begin{lemma}\label{lem:perturb-extreme}
    All extreme sets in $G^\gen$ under edge weights $w$ are also extreme sets under edge weights $w'$.    
\end{lemma}
To show this lemma, we will prove that the (strict) relative order of cut values is preserved by the transformation from $w$ to $w'$. Let $\delta_w(\cdot)$  and $\delta_{w'}(\cdot)$ respectively denote the value of $\delta(\cdot)$ under edge weights $w$ and $w'$. Then, we have the following: 
\begin{lemma}\label{lem:consistency}
    If $\delta_w(X) < \delta_w(Y)$ for two sets of vertices $X, Y\subset V^\gen$, then $\delta_{w'}(X) < \delta_{w'}(Y)$.
\end{lemma}
\begin{proof}
    Since all edge weights are integers, $\delta_w(X) < \delta_w(Y)$ implies
    \begin{equation}\label{eq:integer}
        \delta_w(X) \le \delta_w(Y) - 1.     
    \end{equation}
    Let $r(X)$ (resp., $r(Y)$) denote the sum of the random values $r(u, v)$ over all edges $(u, v)$ that have exactly one endpoint in $X$ (resp., $Y$). Then,
    \begin{align*}
    \delta_{w'}(X) 
    = mN\cdot \delta_w(X) + r(X) 
    &\le mN\cdot (\delta_w(Y) - 1) + r(X) &\text{(by \Cref{eq:integer})}&\\
    &\le mN\cdot \delta_w(Y) & (\text{since } r(u, v)\le N, r(X) \le mN)&\\
    &< mN\cdot \delta_w(Y) + r(Y) = \delta_{w'}(Y). & (\text{since } r(u, v)\ge 1, r(Y) \ge 1) 
    \end{align*}
\end{proof}
We now prove \Cref{lem:perturb-extreme} using \Cref{lem:consistency}:
\begin{proof}[Proof of \Cref{lem:perturb-extreme}]
    Suppose $X$ is an extreme set under edge weights $w$. Then, $\delta_w(Y) > \delta_w(X)$ for all non-empty proper subsets $Y\subset X$. By \Cref{lem:consistency}, this implies that $\delta_{w'}(Y) > \delta_{w'}(X)$. Thus, $X$ is an extreme set under edge weights $w'$ as well.
\end{proof}

\Cref{lem:perturb-extreme} implies that we can use edge weights $w'$ instead of $w$ since our goal is to obtain a tree $T(G^\gen)$ that includes as subtrees all the extreme sets in $G^\gen$ under edge weights $w$. 

We are now ready to describe the recursive algorithm.
There are two base cases: if $|V^\gen| \le 32$ or if $U^\gen = \emptyset$, we use the Bencz\'ur-Karger algorithm~\cite{BenczurK00} to find the extreme sets tree and return it as $T(G^\gen)$.
\eat{
\begin{itemize}
    \item If $|V^\gen| \le 32$, then we use any $O(1)$-time algorithm to find the extreme sets tree and return it as $T(G^\gen)$.
    \item If $U^\gen = \emptyset$, we use the Bencz\'ur-Karger algorithm~\cite{BenczurK00} to find the extreme sets tree and return it as $T(G^\gen)$.
\end{itemize}
}

For the recursive case, we have $|V^\gen| > 32$. Let $s, t$ be two distinct vertices sampled uniformly at random from $V^\gen$ (these vertices may either be contracted or uncontracted vertices), and let $\phi := \con(s, t)$ be the connectivity between $s$ and $t$ in $G^\gen$. We invoke \Cref{thm:ct-alg} on $G^\gen$ to find the \ct $\cut(s, \phi)$ on $G^\gen$ and define $X := V^\gen \setminus \cut(s, \phi)$. We repeat this process until we get an $X$ that satisfies:
\begin{equation}\label{eq:X}
    \frac{|V^\gen|}{16}\le |X|\le \frac{15\cdot |V^\gen|}{16}.
\end{equation} 
Once \Cref{eq:X} is satisfied, we create the following two subproblems:
    \begin{itemize}
        \item In the first subproblem, we contract the vertices in $X$ into a single (contracted) vertex to form a new graph $G^\gen_X$. We find the tree $T(G^\gen_X)$ on $G^\gen_X$ by recursion. 
        \item In the second subproblem, we contract the vertices in $V^\gen\setminus X$ into a single (contracted) vertex to form a new graph $G^\gen_{V^\gen\setminus X}$. We find the tree $T(G^\gen_{V^\gen\setminus X})$ on $G^\gen_{V^\gen\setminus X}$ by recursion. 
    \end{itemize}
    We combine the trees $T(G^\gen_X)$ and $T(G^\gen_{V^\gen\setminus X})$ to obtain the overall tree $T(G^\gen)$ as follows: in tree $T(G^\gen_{V^\gen\setminus X})$, we discard the leaf representing the contracted vertex $V^\gen\setminus X$; let $T_X$ denote this new tree whose leaves correspond to the vertices in $X$. Next, note that $X$ is a contracted vertex in $G^\gen_X$ that appears as a leaf in tree $T(G^\gen_X)$. We replace the contracted vertex $X$ in this tree with the tree $T_X$ to obtain our eventual tree $T(G^\gen)$. (This is illustrated in \Cref{fig:combine}.)

\begin{figure}
    \centering
    \foreach \cap/\i/\clr/\txt in {$T(G_X^{gen})$/0/yellow!80/$X$, $T(G_{V^{gen}\setminus X}^{gen})$/0/blue!80/$V^{gen} \setminus X$} {
\begin{subfigure}{.3\textwidth}
\caption{\cap}
\begin{tikzpicture}
\tikzset{
 roundnode/.style={circle, draw, thin, minimum size=4mm},
 triang/.style={isosceles triangle, draw, anchor=south, shape border rotate=90, thin, minimum size=6mm}
}
 \path
 (1.4+\i,2.6) node[roundnode](root) {}
 (.4+\i,1.6) node[roundnode](a) {}
 (2.4+\i,1.6) node[roundnode](b) {}
 (0+\i,0) node[triang,fill=\clr](c) {}
 (.8+\i,0) node[triang,fill=\clr](d) {}
 (1.6+\i,0) node[triang,fill=\clr](e) {}
 (2.4+\i,0) node[triang,fill=\clr](f) {}
 (3.2+\i,.3) node[roundnode,fill=black](g) {} 
 (3.2+\i,-.3) node {\txt};
 \draw (root)--(a)--(c.north);
 \draw (a)--(d.north);
 \draw (root)--(b)--(e.north);
 \draw (f.north)--(b)--(g);

\end{tikzpicture}
\end{subfigure}
}
\begin{subfigure}{.3\textwidth}
\caption{$T(G^{gen})$}
\begin{tikzpicture}
\tikzset{
 roundnode/.style={circle, draw, thin, minimum size=4mm},
 triang/.style={isosceles triangle, draw, anchor=south, shape border rotate=90, thin, minimum size=6mm}
}
\foreach \label/\i/\j/\clr in {0/0/0/yellow!80,1/1.7/-2.3/blue!80} {
 \path
 (1.2+\i,2.6+\j) node[roundnode](root) {}
 (.35+\i,1.6+\j) node[roundnode](a) {}
 (2.1+\i,1.6+\j) node[roundnode](b\label) {}
 (0+\i,\j) node[triang,fill=\clr](c) {}
 (.7+\i,\j) node[triang,fill=\clr](d) {}
 (1.4+\i,\j) node[triang,fill=\clr](e) {}
 (2.1+\i,\j) node[triang,fill=\clr](f) {};
 \draw (root)--(a)--(c.north);
 \draw (a)--(d.north);
 \draw (root)--(b\label)--(e.north);
 \draw (f.north)--(b\label);
}
\draw(b0)--(root);
\end{tikzpicture}
\end{subfigure}
    \caption{This figure illustrates how the trees obtained from recursive calls $T(G^\gen_X)$ and $T(G^\gen_{V^\gen\setminus X})$ are combined in the first phase of the extreme sets algorithm to obtain the tree $T(G^\gen)$. Here, $s\in X$. Yellow leaves are in $V^{gen}\setminus X$, and blue leaves are in $X$.}
    \label{fig:combine}
\end{figure}
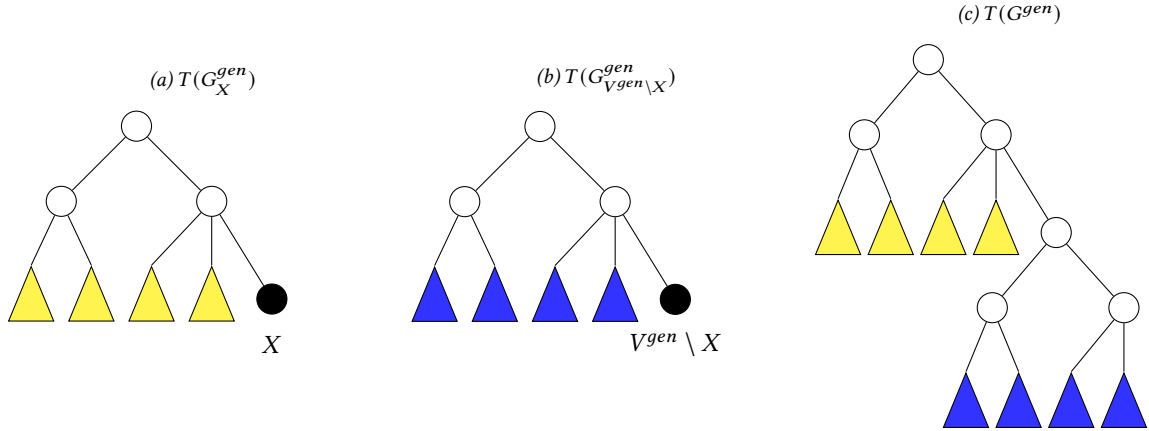

The following is the main claim after the first phase of the algorithm, where $T = T(G)$:
\begin{lemma}\label{lem:first}
    Every extreme set of the input graph $G$ is a subtree of tree $T$ returned by the first phase of the extreme sets algorithm.
\end{lemma}

\paragraph{Phase 2:}
The second phase retains only the subtrees of $T$ that are extreme sets in $G$ and eventually returns $\ext{G}$. In this phase, we do a postorder traversal of $T$. For any vertex $y\in T$, let $V(y)$ denote the set of leaves in the subtree under $y$. During the postorder traversal, we label each vertex $y$ in $T$ with the value of $\delta(V(y))$ in $G$ under the original edge weights $w$. (We will describe the data structures necessary for this labeling when we analyze the running time of the algorithm.) If the label for $y$ is strictly smaller than the labels of all its children nodes, then $V(y)$ is an extreme set and we keep $y$ in $T$. Otherwise, we remove node $y$ from $T$ and make its parent node the new parent of all of its children nodes.

At the end of the second phase of the algorithm, we claim the following:
\begin{lemma}\label{lem:second}
    Every extreme set of the input graph $G$ is a (proper) subtree of tree $T$ returned by the second phase of the extreme sets algorithm, and vice-versa.
\end{lemma}

\subsection{Correctness of the Algorithm}

We now establish the correctness of the algorithm by proving \Cref{lem:first} and \Cref{lem:second} that respectively establish correctness for the first and second phases of the algorithm. 

In order to prove \Cref{lem:first}, we show that the following more general property holds for any recursive step of the algorithm:
\begin{lemma}\label{lem:dandc}
    Let $G^\gen$ be the input graph in a recursive step of the algorithm. Then, every extreme set of $G^\gen$ under edge weights $w$ is a subtree of tree $T(G^\gen)$ returned by the recursive algorithm.
\end{lemma}
Note that \Cref{lem:first} follows from \Cref{lem:dandc} when the latter is applied to the first step of the algorithm, i.e., $G^\gen = G$.

Recall that $X = V\setminus \cut(s, \phi)$, where $s$ is a randomly chosen vertex and $\phi = \con(s, t)$ for a randomly chosen vertex $t\in V\setminus \{s\}$. The two recursive subproblems are on graphs $G^\gen_X$ and $G^\gen_{V^\gen\setminus X}$. To prove \Cref{lem:dandc}, we first relate the extreme sets in $G^\gen_X$ and $G^\gen_{V^\gen\setminus X}$ to the extreme sets in $G^\gen$. We show the following general property that holds for any graph $G^\gen = (V^\gen, E^\gen)$, vertex $s\in V^\gen$, and threshold $\phi\ge 0$:
\begin{lemma}\label{lem:contract}
    Let $G^\gen = (V^\gen, E^\gen)$ be an undirected graph, and for any vertex $s\in V^\gen$ and threshold $\phi\ge 0$, let $X := V^\gen\setminus \cut(s, \phi)$ for \ct $\cut(s, \phi)$ in $G^\gen$ under edge weights $w'$. Let $G^\gen_X$ and $G^\gen_{V^\gen\setminus X}$ be graphs obtained from $G^\gen$ by contracting $X$ and $V^\gen\setminus X$ respectively. Then, every extreme set in $G^\gen$ under edge weights $w$ is an extreme set in either $G^\gen_X$ or $G^\gen_{V^\gen\setminus X}$ under edge weights $w$.
\end{lemma}
\begin{proof}
    First, note that by \Cref{lem:consistency}, every extreme set in $G^\gen$ under edge weights $w$ is also an extreme set under edge weights $w'$. Therefore,
    by applying \Cref{lem:ct-noncross} on $G^\gen$ with edge weights $w'$, we can claim that the extreme sets $Y\subset V^\gen$ under edge weights $w$ are of one of the following types: (a) $Y\subseteq X$ or (b) $X\subseteq Y$ or (c) $X\cap Y = \emptyset$. Extreme sets $Y$ of type (a) are also extreme sets in $G^\gen_{V^\gen\setminus X}$ since the value of $\delta(Y)$ and that of $\delta(Z)$ for any $Z\subset Y$ are identical between $G^\gen$ and $G^\gen_{V^\gen\setminus X}$. Similarly, extreme sets $Y$ of type (c) are also extreme sets in $G^\gen_X$ since the value of $\delta(Y)$ and that of $\delta(Z)$ for any $Z\subset Y$ are identical between $G^\gen$ and $G^\gen_X$. For extreme sets $Y$ of type (b), note that every proper subset of $Y$ in $G^\gen_X$ is also a proper subset of $Y$ in $G^\gen$, and has the same cut value. Then, if $\delta(Z_{G^\gen}) > \delta(Y)$ for all proper subsets $Z_{G^\gen}\subset Y$ in $G^\gen$, then it must be that $\delta(Z_{G^\gen_X}) > \delta(Y)$ for all proper subsets $Z_{G^\gen_X}\subset Y$ in $G^\gen_X$. Therefore, an extreme set of type (b) in $G^\gen$ is also an extreme set in $G^\gen_X$. (Note that because of this last case, it is possible that there are extreme sets in $G^\gen_X$ that are not extreme sets in $G^\gen$.)
\end{proof}
This now allows us to prove \Cref{lem:dandc}:
\begin{proof}[Proof of \Cref{lem:dandc}] 
%
    First, note that the correctness of the base case follows from the correctness of the Bencz\'ur-Karger algorithm~\cite{BenczurK00}.
    %
    Thus, we consider the inductive case. Inductively, we assume that $T(G^\gen_X)$ and $T(G^\gen_{V^\gen\setminus X})$ contain as subtrees all extreme sets of $G^\gen_X$ and $G^\gen_{V^\gen\setminus X}$ under edge weights $w$. Therefore, by \Cref{lem:contract}, every extreme set in $G^\gen$ under edge weights $w$ is a subtree of either $T(G^\gen_X)$ or $T(G^\gen_{V^\gen\setminus X})$. Now, note that any subtree $Y$ eliminated by the algorithm that combines $T(G^\gen_X)$ and $T(G^\gen_{V^\gen\setminus X})$ into $T(G^\gen)$ has the property that $Y$ contains the entire set $V^\gen\setminus X$ and a proper subset of $X$. But, by \Cref{lem:ct-noncross}, such a set $Y$ cannot be an extreme set in $G^\gen$. Therefore, all the extreme sets in $G^\gen$ under edge weights $w$ are subtrees in $T(G^\gen)$.
\end{proof}

\smallskip
Next, we establish correctness of the second phase of the algorithm, i.e., prove \Cref{lem:second}. We will need the following property of extreme sets:
\begin{lemma}\label{lem:subset}
    Let $G = (V, E)$ be an undirected graph, and let $Y \subset V$ be a set of vertices that is {\em not} an extreme set. Then, there exists a set $Z\subset Y$ such that $Z$ is an extreme set and $\delta(Z) \le \delta(Y)$.
\end{lemma}
\begin{proof}
    Let $\xi$ be the minimum cut value among all proper subset of $Y$, i.e., 
    $\xi := \min\{\delta(W): W\subset Y\}$. 
    Since $Y$ is not an extreme set, it must be that $\xi \le \delta(Y)$. Now, consider the smallest set $Z\subset Y$ such that $\delta(Z) = \xi$, i.e., $Z := \arg\min\{|W|: W\subset Z, \delta(W) = \xi\}$. Now, for any non-empty proper subset $R\subset Z$, we have: (a) $\delta(R) \ge \xi$ by definition of $\xi$, and (b) $\delta(R) \not= \xi$ by definition of $Z$. Therefore, $\delta(R) > \xi$ for all non-empty proper subsets $R\subset Z$. Hence, $Z$ is an extreme set.
\end{proof}
We are now ready to prove \Cref{lem:second}:
\begin{proof}[Proof of \Cref{lem:second}]
    Recall that for any node $y$ in $T$, $V(y)\subseteq V$ denotes the set of leaves in the subtree under $y$. Now, if $y$ is removed by the algorithm in the second phase from $T$, it must be that there is a child $z$ of $y$ such that $\delta(V(z)) \le \delta(V(y))$. Since each node in $T$ has at least two children, it must be that $V(z)$ is a proper subset of $V(y)$, and hence $V(y)$ is not an extreme set. This implies that the second phase of the algorithm does not remove any extreme set from being a subtree of $T$. 
    
    It remains to show that this phase {\em does} remove all subtrees that are not extreme sets. Suppose $y$ is a node in $T$ after the first phase of the algorithm such that $V(y)$ is not an extreme set in $G$. Consider the stage when the postorder traversal of $T$ in the second phase reaches $y$. We need to argue that there is a child $x$ of $y$ such that $\delta(V(x))\le \delta(V(y))$. Inductively, we assume that at this stage, the subtree under $y$ exactly represents the extreme sets that are proper subsets of $V(y)$. Then, by \Cref{lem:subset}, there is a descendant $z$ of $y$ such that $\delta(V(z))\le \delta(V(y))$. But, note that in any extreme sets tree, the cut value of a parent subtree is strictly smaller than that of a child subtree, since the child subtree represents a proper subset of the parent subtree. Thus, if $x$ is the child of $y$ that is also an ancestor of $z$, then $\delta(V(x)) \le \delta(V(z)) \le \delta(V(y))$. Since $\delta(V(x))\le \delta(V(y))$ and $x$ is a child of $y$, the node $y$ will be discarded from $T$ when the postorder traversal reaches $y$.
\end{proof}

This concludes the proof of correctness of the extreme sets algorithm.

\subsection{Running Time Analysis of the Algorithm}

We analyze the running times of the first and second phases of the algorithm separately. It follows from \Cref{thm:ct-alg} that the running time $T_1(m, n)$ of the first phase can be written as:
\begin{equation}\label{eq:recurse}
    T(m, n) = T(m_1, n_1) + T(m_2, n_2) + \tO(m) + \polylog(n)\cdot F(m, n),
\end{equation}    
where $n_1 + n_2 = n+2$ and $m_1 + m_2 = m + d(X)$, where $d(X)$ is the number of edges that have exactly one endpoint in $X$. Note that all other steps, i.e., generating edge weights $w'$, creating the graphs $G^\gen_X$ and $G^\gen_{V^\gen\setminus X}$, and recombining the trees $T(G^\gen_X)$ and $T(G^\gen_{V^\gen\setminus X})$ to obtain the overall tree $T(G^\gen)$, can be done in $O(m)$ time. Thus, the running time is dominated by the time taken in the \ct algorithm in \Cref{thm:ct-alg}.

First, we bound the depth of the recursion tree:
\begin{lemma}\label{lem:depth}
    The depth of the recursion tree in the first phase of the extreme sets algorithm is $O(\log n)$.
\end{lemma}
\begin{proof}
   Note that \Cref{eq:X} ensures that in every recursive step, we have:
    \[\max\{|X|, |V^\gen\setminus X|\}\le \frac{15\cdot |V^\gen|}{16}.\]
    Therefore, in each recursive subproblem, the number of vertices is $\le \frac{15\cdot |V^\gen|}{16} + 1 < \frac{31\cdot |V^\gen|}{32}$ since $|V^\gen| > 32$. The lemma follows.
\end{proof}

\Cref{lem:depth} is sufficient to bound the total cost of the base cases of the algorithm:
\begin{lemma}\label{lem:bk}
    The total running time of the invocations of the Bencz\'ur-Karger algorithm for the base cases is $\tO(n)$.
\end{lemma}
\begin{proof}
    First, consider the base cases of constant size: $|V^\gen| \le 32$. Since the other base case truncates the recursion whenever $U^\gen = \emptyset$, it must be that $V^\gen$ contains at least one uncontracted vertex in each invocation of this base case. Now, since each uncontracted vertex is assigned to exactly one of the two subproblems by the recursive algorithm, it follows that each uncontracted vertex can be in only one base case. Therefore, the total number of these bases cases is $\le n$. Since each base case is on a graph of $O(1)$ size, the total running time of the Bencz\'ur-Karger algorithm over these base cases is $O(n)$.
    
    Next, we consider the other base case: $U^\gen  = \emptyset$. Since the depth of the recursion tree is $O(\log n)$ by \Cref{lem:depth}, and each branch of the recursion adds a single contracted vertex in each step, the total number of contracted vertices in any instance is $O(\log n)$. Thus, the Bencz\'ur-Karger algorithm has a running time of $O(\log^2 n\cdot  \polylog(\log n))$ for each instance of this base case. To count the total number of these instances, we note that the parent subproblem of any base case must contain at least one uncontracted vertex. Since the depth of the recursion tree is $O(\log n)$ and an uncontracted vertex can be in only one subproblem at any layer of recursion, it follows that the total number of instances of this base case is $O(n\log n)$. Therefore, the cumulative running time of all the base cases of this type is $\tO(n)$.
\end{proof}

The rest of the proof will focus on bounding the cumulative running time of the recursive instances of the algorithm. 
Our first step is to show that the expected number of iterations in any subproblem before we obtain an $X$ that satisfies \Cref{eq:X} is a constant:
\begin{lemma}\label{lem:X}
    Suppose $s, t$ are vertices chosen uniformly at random from $V^\gen$, and let $\phi := \con(s, t)$ be the $s-t$ connectivity in $G^\gen$. Then, $X := V^\gen \setminus \cut(s, \phi)$ satisfies \Cref{eq:X} with probability $\ge 1/32$.
\end{lemma}
To show this, we first need to establish some properties of the random transformation that changes edge weights from $w$ to $w'$.
First, we establish uniqueness of the minimum $s-t$ cut for any vertex pair $s, t\in V$ under $w'$. We need the {\em Isolation Lemma} for this purpose:
\begin{lemma}[Isolation Lemma~\cite{MulmuleyVV87}]\label{lem:isolation}
    Let $m$ and $N$ be positive integers and let $\cal F$ be a collection of subsets of $\{1, 2, \ldots, m\}$. Suppose each element $x\in \{1, 2, \ldots, m\}$ receives a random number $r(x)$ uniformly and independently from $\{1, 2, \ldots, N\}$. Then, with probability $\ge 1-m/N$, there is a unique set $S\in {\cal F}$ that minimizes $\sum_{x\in S} r(x)$.
\end{lemma}
We choose $N = m\cdot n^d$ for some constant $d > 0$. (Note that this increases the edge weights from $w$ to $w'$ by a $\poly(n)$ factor only, thereby ensuring that the efficiency of elementary operations is not affected.) Then, we can apply the Isolation Lemma to prove the following property:
\begin{lemma}\label{lem:unique}
    Fix any vertex $s\in V^\gen$. For every vertex $t\in V^\gen\setminus \{s\}$, the minimum $s-t$ cut under edge weights $w'$ is unique with probability $\ge 1-1/n^d$. Moreover, let $t, t'\in V^\gen\setminus \{s\}$. With probability at least $\ge 1 - 1/n^d$, one of the following must hold: (a) $\lambda(s, t)\not= \lambda(s, t')$, or (b) the unique minimum $s-t$ cut is identical to the unique minimum $s-t'$ cut in $G^\gen$.
\end{lemma}
\begin{proof}
    We first establish the uniqueness of the minimum $s-t$ cut. Note that by \Cref{lem:consistency}, the only candidates for minimum $s-t$ cut under $w'$ are the minimum $s-t$ cuts under $w$. For any two such cuts $X, Y\subset V^\gen$, we have $\delta_w(X) = \delta_w(Y)$, i.e., $mN\cdot \delta_w(X) = mN\cdot \delta_w(Y)$. Therefore, the $s-t$ minimum cuts under $w'$ are those minimum $s-t$ cuts $X$ under $w$ that have the minimum value of $r(X)$, which is defined as the sum of $r(u, v)$ over all edges $(u, v)$ with exactly one endpoint in $X$. The uniqueness of the minimum $s-t$ cut under edge weights $w'$ now follows from \Cref{lem:isolation} by setting $\cal F$ to the collection of subsets of edges that form the minimum $s-t$ cuts under edge weights $w$.
    
    Next, consider two vertices $t, t'\in V\setminus \{s\}$. If $\con(s, t) \not= \con(s, t')$ under edge weights $w$, assume wlog that $\con(s, t) < \con(s, t')$. This implies that for every $s-t'$ cut $Y$, we have $\delta_w(Y) > \delta_w(X)$, where $X$ is a minimum $s-t$ cut under edge weights $w$. But then, by \Cref{lem:consistency}, we have $\delta_{w'}(Y) > \delta_{w'}(X)$. This implies that $\lambda(s, t) \not= \lambda(s, t')$ under edge weights $w'$. In this case, we are in case (a). Next, suppose $\lambda(s, t) = \lambda(s, t')$ under edge weights $w$. Apply \Cref{lem:isolation} by setting $\cal F$ to be the collection of subsets of edges where each subset forms a minimum $s-t$ cut or a minimum $s-t'$ cut under edge weights $w$. With probability $\ge 1-1/n^d$, we get a unique minimum cut among these cuts under edge weights $w'$. If this unique minimum is a cut that separates both $t, t'$ from $s$, then we are in case (b), while if it only separates one of $t$ or $t'$ from $s$, then we are in case (a).
\end{proof}

Using $d > 3$, and applying a union bound over all choices of $s, t, t'$, we can assume that \Cref{lem:unique} holds for all choices of vertices $s, t, t'$. (This holds {\em with high probability}, which is sufficient for our purpose because our algorithm is Monte Carlo.) 

We also need the following lemma due to Abboud {\em et al.}~\cite{AbboudKT20}:
\begin{lemma}[Abboud {\em et al.}~\cite{AbboudKT20}]\label{lem:tournament}
    Let $G = (V, E)$ be an undirected graph. If $s$ is a vertex chosen uniformly at random from $V$, then with probability $\ge 1/2$, there are $\ge |V|/4$ vertices $t\in V\setminus \{s\}$ such that the {\em $t$-minimal} minimum $s-t$ cut has $\le |V|/2$ vertices on the side of $t$.
\end{lemma}
Here, {\em $t$-minimal} refers to the minimum $s-t$ cut where the side containing $t$ is minimized. But, for our purposes, we do not need this qualification since by \Cref{lem:unique}, the minimum $s-t$ cut in $G^\gen$ is unique under edge weights $w'$.

Now, for any vertex $s$, let $\Con(s)$ denote the sequence of vertices $t\in V^\gen\setminus \{s\}$ in non-increasing order of the value of $\con(s, t)$. (If $\con(s, t) = \con(s, t')$, then the relative order of $t, t'$ in $\Con(s)$ is arbitrary.) We define a {\em run} in this sequence as a maximal subsequence of consecutive vertices that have an identical value of $\con(s, t)$. Combining \Cref{lem:unique} and \Cref{lem:tournament}, we make the following claim:
\begin{lemma}\label{lem:run}
    Let $s$ be a vertex chosen uniformly at random from $V^\gen$. Then, with probability $\ge 1/2$, the longest run in $\Con(s)$ is of length $\le \frac{3|V^\gen|}{4}$.
\end{lemma}
\begin{proof}
    First, note that all vertices $t$ in a run share the same unique minimum $s-t$ cut (and not just the value of $\con(s, t)$) by \Cref{lem:unique}. Thus, if there is a run in $\Con(s)$ has $> \frac{3|V^\gen|}{4}$ vertices, then for all these vertices $t$, the unique minimum $s-t$ cut has $> \frac{3|V^\gen|}{4}$ vertices on the side of $t$. It follows that there are $< \frac{|V^\gen|}{4}$ vertices $t$ that have $\le \frac{|V^\gen|}{2}$ vertices on the side of $t$  in the (unique) minimum $s-t$ cut. The lemma now follows by observing that this can only happens with probability $< 1/2$ by \Cref{lem:tournament} since $s$ is a vertex chosen uniformly at random from $V^\gen$.
\end{proof}
\Cref{lem:run} now allows us to derive the probability of choosing vertices $s$ and $t$ such that \Cref{eq:X} is satisfied:
\begin{proof}[Proof of \Cref{lem:X}]
    By \Cref{lem:run}, the longest run in $\Con(s)$ is of length $\le \frac{3|V^\gen|}{4}$ with probability $\ge 1/2$. Next, the index of $t$ in $\Con(s)$ is between $\frac{7 |V^\gen|}{8}$ and $\frac{15 |V^\gen|}{16}$ with probability $1/16$ since $t$ is chosen uniformly at random. If this happens, then we immediately get $|V^\gen\setminus X| = |\cut(s, \phi)| \ge \frac{|V^\gen|}{16}$ where $\phi = \con(s, t)$. This is because the suffix of $\Con(s)$ starting at $t$ is in $\cut(s, \phi)$. But, we also have $|V^\gen\setminus X| = |\cut(s, \phi)|\le \frac{|V^\gen|}{8} + \frac{3|V^\gen|}{4} = \frac{7|V^\gen|}{8}$ since the longest run in $\Con(s)$ has $\le \frac{3|V^\gen|}{4}$ vertices, and all vertices before the start of the run containing $t$ are not in $\cut(s, \phi)$. The lemma follows.
\end{proof}

Next, we bound the total number of vertices and edges at any level of the recursion tree:
\begin{lemma}\label{lem:layer}
    The total number of vertices and edges in all the recursive subproblems at any level of the recursion tree in the first phase of the extreme sets algorithm is $O(n\log n)$ and $O(m + n\log^2 n)$ respectively.
\end{lemma}
\begin{proof}
    Since each step of the recursion adds one contracted vertex to each of the two subproblems, it follows from \Cref{lem:depth} that any subproblem in the recursion tree has at most $O(\log n)$ contracted vertices, i.e., $|C^\gen| = O(\log n)$. Next, note that every uncontracted vertex belongs to exactly one subproblem at any level of the recursion tree. Conversely, because of the base case for $U^\gen = \emptyset$, every recursive subproblem contains at least one uncontracted vertices. Therefore, the recursive subproblems at any level of the recursion tree contain $\le n$ uncontracted vertices and $O(n\log n)$ contracted vertices in total.
    
    The edges in a subproblem are in three categories: (a) edges between two uncontracted vertices, i.e., $\{(u, v)\in E^\gen: u, v\in U^\gen\}$ (b) edges between contracted and uncontracted vertices, i.e., $\{(u, v)\in E^\gen: u\in C^\gen, v\in U^\gen\}$ and (c) edges between two contracted vertices, i.e., $\{(u, v)\in E^\gen: u, v\in C^\gen\}$. Edges in (a) are distinct between subproblems at any level of the recursion tree since the sets of uncontracted vertices $U^\gen$ in these subproblems are disjoint. An edge $(u, v)\in E$ can appear in at most two subproblems as a category (b) edge, namely the subproblems containing the uncontracted vertices $u$ and $v$ respectively. As a result, there are $O(m)$ edges of category (a) and (b) in total across all the subproblems at a single level of the recursion tree. Finally, since the number of contracted vertices is $O(\log n)$ in any single subproblem, there are at most $O(\log^2 n)$ edges of category (c) in any subproblem. Since each recursive subproblem contains at least one uncontracted vertex, the total number of subproblems in a single layer of the recursion tree is $\le n$. Consequently, the total number of edges in category (c) across all subproblems at a single level of the recursion tree is $O(n\log^2 n)$.
\end{proof}

This lemma allows us to bound the running time of the first phase of the algorithm:
\begin{lemma}\label{lem:runtime1}
    The expected running time of the first phase of the algorithm is $\tO(m) + \polylog(n)\cdot F(m, n)$, where $F(m, n)$ is the running time of a max-flow algorithm on an undirected graph of $n$ vertices and $m$ edges.
\end{lemma}
\begin{proof}
    We have already shown a bound of $\tO(n)$ on the base cases in \Cref{lem:bk}. So, we focus on the recursive subproblems. 
    Cumulatively, over the recursive subproblems at a single level, \Cref{lem:layer} asserts that the total number of vertices and edges is $\tO(n)$ and $\tO(m)$ respectively. (Note that we can assume w.l.o.g. that $G$ is a connected graph and therefore $O(n\log^2 n) = \tO(m)$. If $G$ is not connected, we run the algorithm on each connected component separately.) Now, since $\tO(m) + F(m, n) = \Omega(m)$, the total time at a single level of the recursion tree is maximized when there are $\polylog(n)$ subproblems containing $n$ vertices and $m$ edges each. This gives a total running time bound of $\tO(m) + \polylog(n)\cdot F(m, n)$ on the subproblems at a single level. (Note that by \Cref{lem:X}, the expected number of choices of $s, t$ before \Cref{eq:X} is satisfied is a constant.) The lemma now follows by \Cref{lem:depth} which says that the number of levels of the recursion tree is $O(\log n)$.
\end{proof}

Next, we analyze the running time of the second phase of the algorithm.
To implement the second phase, we need to find the value of $\delta(X)$ for all subtrees of $T$. We use a dynamic tree data structure for this purpose. Initialize $\cnt[X] := 0$ for all subtrees $X$. For every edge $(u, v)\in E$, we make the following changes to $\cnt$:
\begin{itemize}
    \item Increase $\cnt[z]$ by $w(u, v)$ for all ancestors $z$ of $u$ and $v$ in $T$.
    \item Decrease $\cnt[z]$ by $2 w(u, v)$ for all ancestors $z$ of $\lca(u, v)$ in $T$.
\end{itemize}
Clearly, the value of $\cnt$ at the end of these updates is equal to $\delta(X)$ for every subtree $X$. Recall that during the postorder traversal for subtree $X$, we declare it to be an extreme set if and only if the value $\cnt[X]$ is {\em strictly} smaller than that of each of its children subtrees. 

This implementation of the second phase of the algorithm gives the following:
\begin{lemma}\label{lem:runtime2}
    The second phase of the extreme sets algorithm takes $\tO(m)$ time.
\end{lemma}
\begin{proof}
    First, note that the size of the tree $T$ output by the first phase is $O(n)$ since the leaves exactly correspond to the vertices of $G$. Thus, the number of subtrees of $T$ is also $O(n)$. The initialization of the dynamic tree data structure takes $O(n)$ time. Then, each dynamic tree update takes $O(\log n)$ time, and there are $O(m)$ such updates. So, the overall time for dynamic tree operations is $\tO(m)$. Finally, the time spent at a node of $T$ during postorder traversal is proportional to the number of its children, which adds to a total time of $O(n)$ for postorder traversal of $T$.
\end{proof}

\section{Augmentation on Extreme Sets}
\label{sec:augment}

In this section, we present our algorithm for degree-constrained edge connectivity augmentation (\dc) that uses extreme sets as a subroutine. Our goal is to prove \Cref{thm:deca}, restated below.
\Deca*

Throughout, we specify a \dc instance by a tuple $(G,\tau,\beta)$, indicating the graph $G$, the connectivity requirement $\tau$, and the (weighted) degree constraints $\beta(v)\ge 0$ for each vertex $v$.

\subsection{The Bencz\'ur-Karger Algorithm for \dc}
As mentioned before, our algorithm is essentially a speedup of the Bencz\'ur-Karger algorithm for \dc~\cite{BenczurK00} from $\tO(n^2)$ time to $\tO(m)$ given the extreme sets tree. We first describe the Bencz\'ur-Karger algorithm and then describe our improvements.

The algorithm consists of 3 phases.
\begin{enumerate}
 \item Using {\em external augmentation}, transform the degree constraints $\beta(v)$ to {\em tight} degree constraints $b(v)$ for all $v\in V$.
 \item Repeatedly add an augmentation chain to increase connectivity to at least $\tau -1$.
 \item Add a matching defined on the min-cut cactus if the connectivity does not reach $\tau$.
\end{enumerate}

We first describe the external augmentation problem and an algorithm (from \cite{BenczurK00}) to optimally solve it.

\paragraph{External augmentation.} 
The problem is defined as follows: Given a \dc instance $(G,\tau,\beta)$, insert a new node $s$, and find an edge set $F\subseteq \{s\}\times V$ with minimum total weight such that $\forall U\subset V$, $\delta_G(U) + \delta_F(U)\ge \tau$, and $\forall v\in V, d_F(v)\le \beta(v)$, where
$d_F(v) := \sum_{u\in V} w_F(u, v)$ is the (weighted) degree of $v$ in edges $F$.

The external augmentation problem can be solved using the following algorithm (from \cite{BenczurK00}): Let $b(v)$ denote the degree of $v$ in new edges. For any set $X\subseteq V$, let $b(X) := \sum_{v\in X} b(v)$. Initially, $b(v)=0$ for all $v\in V$. We do a postorder traversal on the extreme sets tree. When visiting an extreme set $X$ that is still deficient, i.e., $b(X) < \dem(X) := \max(\tau - \delta_G(X), 0)$, we add edges from vertices $v\in X$ with $b(v) < \beta(v)$ to $s$ until $b(X)=\dem(X)$. When we fail to find a vertex $v\in X$ such that $b(v) < \beta(v)$, the \dc instance is infeasible since we have $\delta(X)+\beta(X)<\tau$. This algorithm can be implemented in $O(n)$ time using a linked list to keep track of vertices $v$ with $b(v) < \beta(v)$ in a subtree, merging these lists as we move up the tree in the postorder traversal and removing vertices once $b(v) = \beta(v)$.

\begin{lemma}[Lemma 3.4 and 3.6 of \cite{BenczurK00}]
The algorithm described above outputs an optimal solution for the external augmentation problem.  
\end{lemma}

The next lemma (from \cite{BenczurK00}) relates optimal solutions of the external augmentation and \dc problem instances:
\begin{lemma}[Lemma 2.6 of \cite{BenczurK00}]\label{lem:bk2.6}
If the optimal solution of the external augmentation instance has total weight $w$, then the optimal solution of \dc instance has value $\lceil w/2 \rceil$.
\end{lemma}

After external augmentation, we have $w = b(V)$. If $w$ is odd, we claim there is at least one vertex with $\beta(v) \ge b(v)+1$, else the instance is infeasible. \Cref{lem:bk2.6} claims that the optimal solution of the \dc instance has weight $(w+1)/2$, i.e., the sum of degrees is $w+1$. Now, if $\beta(v) = b(v)$ for all vertices $v\in V$, then $\sum_{v\in V}\beta(v) = b(V) = w$. This shows that the instance is infeasible. If the instance is feasible, we add 1 to $b(v)$ for an arbitrary vertex $v\in V$ such that $\beta(v) \ge b(v)+1$. 

By \Cref{lem:bk2.6}, the optimal solution of \dc problem has $b(V)/2$ edges. Now, note that if we had used $b$ instead of $\beta$ as our degree constraints, we would still get the same external augmentation solution and consequently the same value of $w$. Therefore, we call $b$ the {\em tight} degree constraints. The \dc problem is now equivalent to splitting off the vertex $s$ on the external augmentation solution $H = (V+s, E\cup E_s)$ where $E_s$ is the set of weighted edges incident on $s$ where $w(v, s) = b(v)$. We denote this splitting off instance $(H, \tau, s)$.

The Bencz\'ur-Karger algorithm~\cite{BenczurK00} provides an iterative greedy solution for splitting off $s$ by using \emph{partial solutions}. Given a splitting off instance $(H = (V+s, E\cup E_s),\tau, s)$ where $w(v, s) = b(v)$ for all $v\in V$, define a \emph{partial solution} as an edge set $F$ defined on $V$ satisfying the following three properties:
 \begin{enumerate}
 \item For all vertices $v\in V$, the (weighted) degree of $v$ in edges $F$, denoted $d_F(v) := \sum_{u\in V} w_F(u, v)$, satisfies $d_F(v)\le b(v)$.
 \item For all edges $(u,v)\in F$, no extreme set can contain both $u$ and $v$. (Note that an extreme set is a proper subset of $V$, and hence $V$ is not an extreme set by definition.)
 \item Any extreme set in $(V, E\uplus F)$ is also extreme in $G$. (For two weighted edge sets $X, Y$ defined on $V$, we use $X\uplus Y$ to denote their union where the weights of parallel edges are added.) That is, adding $F$ to $G$ does not create new extreme sets (but some extreme sets may no longer be extreme).
 \end{enumerate}


The next lemma shows the optimality of iteratively adding partial solutions for the splitting off problem:
\begin{lemma}[Lemma 4.1 of \cite{BenczurK00}]
\label{lem:partial-solution-split}
Suppose we are given a splitting off instance $(H = (V+s, E\cup E_s),\tau,s)$ and a partial solution $F$ where $d_F(v) := \sum_{u\in V} w_F(u, v)$ is the degree of vertex $v\in V$ in $F$. Now, suppose $F'$ is a solution for the splitting off instance $(H' = (V+s, E'\cup E'_s), \tau, s)$ where the weight of edges in $E'$ and $E'_s$ are respectively given by $w'(u, v) = w(u, v) + w_F(u, v)$ for $u, v\in V$ and $w'(v, s) := w(v, s) - d_F(v)$ for $v\in V$. Then, $F\uplus F'$ is a solution for the splitting-off instance $(G = (V+s, E),\tau,s)$.
\end{lemma}

By equivalence between the splitting off problem and edge augmentation with tight degree constraints $b(v)$, we get the following equivalent lemma for the \dc problem:
%
\begin{lemma}
\label{lem:partial-solution}
Given a \dc instance $(G,\tau,b)$ with tight degree constraints $b$ and given a partial solution $F$, if $F'$ is an optimal solution for \dc instance $(G,\tau,b')$ where $b'(v)=b(v)-d_F(v)$, then $F\uplus F'$ is an optimal solution for the original instance. 
\end{lemma}

For an extreme set $X$ of a graph $G$, define its \emph{demand} as $\dem_G(X)=\tau-\delta_G(X)$. Note that if each extreme set has demand at most $0$, then the graph has connectivity at least $\tau$; this is because there exists a side of a global min-cut (in particular, any {\em minimal} vertex set that is a side of a global min-cut) which is an extreme set.

Consider all maximal extreme sets $X$ satisfying $\dem(X)\ge2$. List them out as $X_1,\ldots,X_r$, where the ordering is such that $X_1$ and $X_r$ have the two smallest values of $\delta_G(X_i)$ among $X_1,\ldots,X_r$.
An \textit{augmentation chain} is an edge set $\{(a_i,\tilde a_{i+1})\mid i\in[r-1]\}$ such that for each $i\in[r-1]$,
\begin{enumerate}
\item $a_i\in X_i$ and $\tilde a_{i+1}\in X_{i+1}$, i.e., edge $(a_i,\tilde a_{i+1})$ connects adjacent sets $X_i$ and $X_{i+1}$, and
\item $b(a_i)\ge d_F(a_i)$ and $b(\tilde{a}_i)\ge d_F(\tilde{a}_i)$ (we say $a_i$ and $\tilde a_i$ still has {\em vacant} degree). Note that $d_F(a_i)=1$ if $a_i\ne \tilde{a}_i$ (or if either $a_i$ or $\tilde a_i$ is undefined), and $d_F(a_i)=2$ if $a_i= \tilde{a}_i$.
\end{enumerate}
The significance of an augmentation chain is that it is always a partial solution. The lemma below is proved in Section~4.2 of~\cite{BenczurK00} and is one of the main technical contributions of that paper.

\begin{lemma}
An augmentation chain is a partial solution.
\end{lemma}

Bencz\'ur and Karger's algorithm repeatedly constructs augmentation chains until there are no extreme sets with demand at least $2$ in the current graph. By applying \Cref{lem:partial-solution} after each iteration, any optimal solution to the instance after that iteration can be augmented to an optimal solution to the instance before that iteration. At the end, only extreme sets with demand $1$ remain in the instance, at which point Bencz\'ur and Karger use an algorithm of Naor~et~al.~\cite{NaorGM97} that runs in $O(n)$ time given the \emph{min-cut cactus representation} of $G$. Using the $\tilde{O}(m)$-time min-cut cactus algorithm of Karger and Panigrahi~\cite{KargerP09}, this last step takes $\tilde{O}(m)$ time.

On each iteration, Bencz\'ur and Karger compute an augmentation chain from scratch given the current extreme sets tree, which takes $O(n)$ time, and then augment with that chain for as long as it is feasible. In particular, they repeatedly augment until some vertex uses up its vacant degree, or the list $X_1,\ldots,X_r$ changes, which can happen in any of the following ways:
 \begin{enumerate}
 \item Some vertex $u$ in the chain has no more vacant degree.
 \item Some $X_i$'s demand decreases to below $2$, in which case it is removed from the list. \label{case:1}
 \item Some $X_i$ is no longer extreme, in which case we replace $X_i$ with the maximal extreme sets in the subtree rooted at $X_i$ of the (original) extreme set tree.\label{case:2}
 \item $X_1$ and $X_r$ are no longer the two extreme sets with smallest cut value in the current graph. Since $X_1$ and $X_r$ have their cut values increased by $1$ on each augmentation while the other extreme sets $X_2,\ldots,X_{r-1}$ have their cut values increased by $2$, this can never happen on its own. In particular, it can only happen alongside cases (\ref{case:1}) and (\ref{case:2}).
 \end{enumerate}
The algorithm therefore computes the minimum number of times $t(F)$ that an augmentation $F$ (i.e., a chain) can be added to the current graph. We can compute $t(F)$ as $\min\{t_1(F),t_2(F),t_3(F)\}$ where each $t_i(F)$ is the time of violation of the respective case above. In particular,
\begin{align*}
 t_1(F) &= \min_{u\in V} \left\lfloor \frac{b(u)}{d_F(u)}\right\rfloor.\\
 t_2(F) &= \min_{i\in[r]} \left\lfloor \frac{\tau-\delta(X_i)}{\delta_F(X_i)}\right\rfloor.\\
 t_3(F) &= \min_{i\in[r]} \min_{Y\in\text{desc}(X_i)}\left\lceil \frac{\delta(Y)-\delta(X_i)}{\delta_F(X_i)-\delta_F(Y)}\right\rceil,\\
\end{align*}
where desc$(U)$ is the set of descendants of $U$ (excluding $U$) in the (original) extreme sets tree.
Note that $d_F(u)$ and $\delta_F(U)$ can be either 1 or 2, and $\delta_F(X_i)\ge \delta_F(Y)$ for all $Y\in \text{desc}(U)$. Also, if the denominator of any of the above fractions is $0$, then we can ignore that expression in the minimum computation.

\subsection{Improved Algorithm}
We now speed up the Bencz\'ur-Karger algorithm so that it takes $O(n\log n)$ time given the extreme sets tree of the input graph (except the last step that uses a min-cut cactus and takes $\tO(m)$ time). Our main insight is the following: rather than computing each new augmentation chain from scratch, we want to reuse as many edges from the previous chain as possible. We show that any changes that must be made can be amortized to a total of $O(n\log n)$ time with the help of data structures. Our speedup changes can be summarized as follows:
 \begin{itemize}
 \item We maintain $t_1(F),t_2(F),t_3(F)$ using data structures so that $t(F)$ can be computed quickly at each iteration, and
 \item Instead of adding each augmentation chain explicitly to the graph in $O(n)$ time, we add it implicitly with the help of ``lazy'' tags on each edge.
 \end{itemize}

\subsubsection{Data Structure}
For $t_1(F)=\min_{u}  \lfloor b(u)/d_F(u)\rfloor$, we only need to consider vertices $u\in\{a_i,\tilde a_{i+1}\}$ for some $i$, since those are the only vertices with $d_F(u)>0$. Since only $d_F(u)\in\{1,2\}$ is possible for such $u$, we use two priority queues maintaining $b(u)$ for $d_F(u)=1$ and $d_F(u)=2$. Modifying $d_F(u)$ can be handled by deleting $u$ from one queue and insert it to the other. Other (single element) operations can be handled by normal priority queue operations in $O(\log n)$ time.
Call this data structure the \emph{dual} priority queue. Let $Q_1$ be the dual priority queue used to maintain $t_1(F)$, and let $Q_1[1]$ and $Q_1[2]$ be the two priority queues responsible for $d_F(u)=1$ and $d_F(u)=2$, respectively. Similarly, $t_2(F)$ can be maintained by a dual priority queue $Q_2$ since $\delta_F(X_i)\in\{1,2\}$ for all $i\in[r]$, and define $Q_2[1]$ and $Q_2[2]$ as before. Maintaining $t_3(F)$ is more involved, so we defer its discussion to later.

We maintain the edges $(a_i,\tilde{a}_{i+1})\in F$ and the list $X_1,\ldots,X_r$ explicitly. The function $b$ is implicitly maintained by $Q_1$, and values $\delta(Y)$ are implicitly maintained by $Q_2$ and $R(X_i)$. To maintain these implicitly, we keep a global ``timer'' $t_{global}$ that starts at $0$ and increases by $t(F)$ every time we add the current augmentation chain $F$ to the graph. Every time some edge $e$ is added to $F$, we maintain the edge's ``birth'' time $t_{birth}(e)$ which we set to the current global timer $t_{global}$. At any later point in time, if edge $e$ is still in $F$, then its weight is implicitly set to $t_{global}-t_{birth}(e)$. The moment an edge $e$ is removed from $F$, we explicitly add an edge $e$ of weight $t_{global}-t_{birth}(e)$ to the current graph. Similarly, every time a vertex $u$ has an incident edge added or removed from $F$, we set its birth time $t_{birth}(u)$ to the current $t_{global}$.

We now discuss how to implicitly maintain $b$ and $\delta$ in the dual priority queues $Q_1,Q_2$. Every time we add or delete an edge $e$ in $F$, we update $Q_1$ as follows. For each endpoint $u$ of $e$ whose value $d_F(u)$ after the modification is positive, we add $u$ to the priority queue of $Q_1[d_F(u)]$ and set its value to $b(u)-d_F(u)\cdot t_{birth}(u)$. (If $u$ already existed in $Q_1$ before, then delete it before inserting it again.) This way, we maintain the invariant that at any later time $t_{global}$, the true value of $b(u)$ is exactly $u$'s value in $Q_1$ plus $d_F(u)\cdot t_{global}$. The key observation is that for a given $t_{global}$ and a given $i\in\{1,2\}$, the true values $b(u)$ for each vertex $v$ in $Q_1[i]$ are off from their $Q_1$ values by the same additive $i\cdot t_{global}$. Therefore, by querying the minimum in $Q_1[i]$ for $i\in\{1,2\}$, we can recover the correct minimum $t_1(F)=\min_{u}  \lfloor b(u)/d_F(u)\rfloor$.

Similarly for $Q_2$, every time we add/delete an edge in $F$ that connects $X_i$ and $X_{i+1}$, we move each $X\in\{X_i,X_{i+1}\}$ to $Q_2[\delta_F(X)]$ and set its value to be its old value minus $t_{birth}(e)$ in the case of addition, and its old value plus $t_{global}$ in the case of deletion. After deletion, the edge $e$ has been explicitly added with weight $t_{global}-t_{birth}(e)$, which is exactly the net contribution over the insertion and deletion. Once again, for a given $t_{global}$ and a given $i\in\{1,2\}$, the true values $\delta_F(X_i)$ in $Q_2[i]$ are off from their $Q_2$ values by the same additive $i\cdot t_{global}$, so querying the minimum in $Q_2[i]$ for $i\in\{1,2\}$ lets us recover $t_2(F)=\min_{i\in[r]} \left\lfloor (\tau-\delta(X_i))/\delta_F(X_i)\right\rfloor$.

We now discuss how to maintain $t_3(F)$. We maintain values $t_3(F,X_i)=\displaystyle\min_{Y\in\text{desc}(X_i)}\left\lceil \frac{\delta(Y)-\delta(X_i)}{\delta_F(X_i)-\delta_F(Y)}\right\rceil$ for each $X_i$ in the current list $X_1,\ldots,X_r$. The value $t_3(F,X_i)$ is first computed when we add a new $X_i$ to the list, and it is updated whenever we add or remove an edge in $\delta(X_i)$, or we swap $X_i$ in the ordering (in particular, when a different extreme set becomes $X_1$ or $X_r$). From the values $t_3(F,X_i)$, we can easily maintain $t_3(F)=\min_{i\in[r]}t_3(F,X_i)$ using a priority queue whose values are the $t_3(F,X_i)$.

To maintain the values $t_3(F,X_i)$, we use a (static) tree data structure that maintains a real number at each vertex of the tree and supports the following operations, which can be implemented in $O(\log n)$ amortized time by, e.g., a top tree (see Section~6 of~\cite{goldberg1991use}).
 \begin{itemize}
 \item $\textsc{AddPath}(u,v,x)$: add real number $x$ to all vertices on the $u-v$ path in the tree,
 \item $\textsc{MinPath}(u,v)$: return the minimum value of all vertices on the $u-v$ path in the tree, and
 \item $\textsc{MinSubtree}(u)$: return the minimum value of all vertices in the subtree rooted at $v$.
 \end{itemize}
Our static tree is just the original extreme set tree itself, whose nodes are the extreme sets of the original graph. For each extreme set $X$ in the original graph, we implicitly maintain the value $\delta(X)$ at node $X$ in the tree. Every time an edge $(u,v)$ of weight $w$ is explicitly added to the graph (i.e., when it is removed from $F$), we explicitly update the values $\delta(X)$. Let $Y$ be the lowest common ancestor of extreme sets $\{u\}$ and $\{v\}$ in the tree. The extreme sets that contain edge $(u,v)$ are precisely those on the $\{u\}$-to-$\{v\}$ path in the tree, excluding $Y$. We can therefore call $\textsc{AddPath}(\{u\},\{v\},w)$ and $\textsc{AddPath}(Y,Y,-w)$ to explicitly update the values $\delta(X)$.

Of course, to compute $t_3(F)$, we also need to consider the edges implicitly added to the graph, i.e., the edges currently in $F$. We first assume that $1<i<r$. For each such $X_i$, let $Y_i$ be the lowest common ancestor of extreme sets $\{a_i\}$ and $\{\tilde a_i\}$. Then, observe that
 \begin{itemize}
 \item The extreme sets $Y\in\text{desc}(X_i)$ with $\delta_F(Y)=2$ are precisely those on the path from $Y_i$ to $X_i$, excluding $X_i$,
 \item The extreme sets $Y\in\text{desc}(X_i)$ with $\delta_F(Y)=1$ are precisely those on the path from $\{a_i\}$ to $\{\tilde a_i\}$, excluding $Y_i$, and
 \item All other extreme sets $Y\in\text{desc}(X_i)$ satisfy $\delta_F(Y)=0$.
 \end{itemize}
We compute the minimum $\delta(Y)$ conditioned on $\delta_F(Y)=0,1,2$ separately. We first call $\textsc{AddPath}(X_i,X_i,M)$ for a large value $M>0$ so that $X_i$ is no longer the minimum in any of our \textsc{MinPath} queries. For $\delta_F(Y)=2$, we call $\textsc{MinPath}(\{Y_i\},\{X_i\})$ and add the implicit weights of the edges incident to $a_i$ and $\tilde a_i$ in $F$. For $\delta_F(Y)=1$, we call $\textsc{MinPath}(\{a_i\},Y_i)$ and add the implicit weight of the edge incident to $a_i$ in $F$, then call $\textsc{MinPath}(\{\tilde a_i\},Y_i)$ and add the implicit weight of the edge incident to $\tilde a_i$ in $F$, and finally take the minimum of the two. For $\delta_F(Y)=0$, we call $\textsc{AddPath}(\{Y_i\},\{X_i\},M)$ and $\textsc{AddPath}(\{a_i\},\{\tilde a_i\},M)$ to exclude those extreme sets from the minimum computation, and then call $\textsc{MinSubtree}(X_i)$. Finally, we reverse all the \textsc{AddPath} queries by calling them again with $-M$ instead of $M$. The case $i\in\{1,r\}$ is handled similarly.

With the help of the tree data structure, we can also compute the new list $X_1,\ldots,X_r$ whenever a set $X_i$ is removed from it, i.e., when case~(\ref{case:1}) or~(\ref{case:2}) happens. Whenever a set $X_i$ is removed, we traverse down the subtree rooted at $X_i$ to determine the maximal extreme sets in the subtree with demand at least $2$. To determine whether a set $Y$ is still extreme, we compute $\min_{Y'\in\text{desc}(Y)}\delta(Y')$ by casing on the value of $\delta(Y')\in\{0,1,2\}$ in the same way as above, and comparing its value to $\delta(Y)$. Whenever we find an extreme set $Y$ with demand at least $2$, we stop traversing down the subtree at $Y$ and look elsewhere.

\subsubsection{Running Time}
We claim that the running time of our algorithm is $O(n\log n)$ given the original extreme sets tree. Recall that each iteration stops when one of the following occurs.
 \begin{enumerate}
 \item Some vertex $u$ has no more vacant degree. In this case, we replace the edges incident to $u$ in $F$, which is at most $2$ edges. This takes $O(\log n)$ time, and this case can happen at most $n$ times, once per vertex.
 \item Some $X_i$'s demand decreases to below $2$, or some $X_i$ is no longer extreme. In this case, we remove $X_i$ from the list and add the maximal extreme sets with demand at least $2$ in the subtree rooted at $X_i$ in the original extreme set tree. The algorithm traverses down the subtree rooted at $X_i$ to look for the new extreme sets to add to the list. Here, the key observation is that each extreme set in the original extreme set tree is visited at most once. Once it is visited in one of these traversals, it is either verified to be extreme with demand at least $2$, in which case it is added to the list, or not, in which case it is never visited again. Therefore, the total number of extreme sets to be verified is $O(n)$ over the iterations. Each verification takes $O(\log n)$ time for a total of $O(n\log n)$.

As for edge modifications, there are at most $2$ edge modifications each time some $X_i$ is added or removed from the list. Each extreme set is added and removed at most once, for a total of $O(n)$ modifications over the iterations. We only explicitly add edges to the graph after each such modification, and updating the data structures on each addition takes $O(\log n)$ time for a total of $O(n\log n)$.

 \end{enumerate}
Including the last step that uses the min-cut cactus and takes $\tO(m)$ time, the total running time is $\tO(m)$, which concludes \Cref{thm:deca}.

\section*{Acknowledgements} 

RC and DP were supported in part by an NSF CAREER award CCF-1750140, NSF award CCF-1955703, and ARO award W911NF2110230. JL was supported in part by NSF awards CCF-1907820, CCF-1955785, and CCF-2006953.

\bibliographystyle{plain}
\bibliography{ref}

\end{document}